\documentclass[submission,copyright,creativecommons]{eptcs}
\providecommand{\event}{QPL 2019}

\usepackage[utf8]{inputenc}
\usepackage{scrextend}
\usepackage[english]{babel}
\usepackage{amsthm}
\usepackage{stmaryrd}
\usepackage{graphicx}
\usepackage{keycommand}
\usepackage{hyperref}
\usepackage[all]{hypcap}

\theoremstyle{definition}
\newtheorem{theorem}{Theorem}[section]

\newtheorem{lemma}[theorem]{Lemma}
\newtheorem{proposition}[theorem]{Proposition}

\newtheorem{example*}[theorem]{Example*}
\newtheorem{examples*}[theorem]{Examples*}

\newtheorem{remark*}[theorem]{Remark*}

\newtheorem*{theorem*}{Theorem}
\newtheorem*{corollary*}{Corollary}
\newtheorem*{lemma*}{Lemma}
\newtheorem*{proposition*}{Proposition}

\usepackage{tikzit}
\input{zh.tikzdefs}

\tikzstyle{dot}=[inner sep=0.3mm, minimum width=2mm, minimum height=2mm, draw, shape=circle, font={\footnotesize}, tikzit fill=magenta]
\tikzstyle{white dot}=[dot, fill=white, text depth=-0.2mm, tikzit category=ZH-pf]
\tikzstyle{gray dot}=[dot, fill={rgb,255: red,128; green,128; blue,128}, text depth=-0.2mm, tikzit category=ZH-pf]
\tikzstyle{gray phase dot}=[gray dot, fill={rgb,255: red,128; green,128; blue,128}, tikzit fill=magenta]
\tikzstyle{hadamard}=[fill=white, draw, inner sep=0.6mm, minimum height=1.5mm, minimum width=1.5mm, shape=rectangle, tikzit shape=rectangle, tikzit category=ZH-pf]
\tikzstyle{small hadamard}=[fill=white, draw, inner sep=0.6mm, minimum height=1.5mm, minimum width=1.5mm, tikzit shape=rectangle]
\tikzstyle{halfscalar}=[star, fill=black, draw=black, minimum size=8pt, inner sep=0pt]
\tikzstyle{box}=[shape=rectangle, text height=1.5ex, text depth=0.25ex, yshift=0.2mm, fill=white, draw=black, minimum height=3mm, minimum width=5mm, font={\small}]
\tikzstyle{Z dot}=[inner sep=0mm, minimum size=2mm, shape=circle, draw=black, fill=zx_green]
\tikzstyle{Z phase dot}=[minimum size=5mm, font={\footnotesize\boldmath}, shape=rectangle, rounded corners=2mm, inner sep=0.2mm, outer sep=-2mm, scale=0.8, tikzit shape=circle, draw=black, fill=zx_green, tikzit draw=blue]
\tikzstyle{X dot}=[Z dot, shape=circle, draw=black, fill=zx_red]
\tikzstyle{X phase dot}=[Z phase dot, tikzit shape=circle, tikzit draw=blue, fill=zx_red, font={\footnotesize\color{black}\boldmath}]
\tikzstyle{zxnode}=[inner sep=0.3mm, minimum width=2mm, minimum height=2mm, draw, shape=circle, font={\footnotesize}]
\tikzstyle{gn}=[zxnode, fill=zx_green]
\tikzstyle{rn}=[zxnode, fill=zx_red]
\tikzstyle{H box}=[rectangle, fill=yellow, draw=black, xscale=1, yscale=1, font={\small}, inner sep=0.75pt, minimum width=0.15cm, minimum height=0.15cm, tikzit shape=rectangle]
\tikzstyle{ug}=[regular polygon, regular polygon sides=3, fill=zx_red, draw=black, inner sep=0pt, minimum width=1em, tikzit draw=blue]
\tikzstyle{st}=[star, star points=5, fill=white, draw=black, inner sep=1.2pt, line width=1.2pt, tikzit fill=blue, tikzit draw=red, tikzit category=ZH-pf]
\tikzstyle{triangle}=[regular polygon, regular polygon sides=3, fill=white, draw=black, inner sep=0pt, minimum width=1em, tikzit draw=blue, tikzit category=ZH-pf]
\tikzstyle{not}=[fill={rgb,255: red,128; green,128; blue,128}, draw=black, shape=circle, font={$\neg$}, dot]

\tikzstyle{gray}=[-, draw={blue!60!white}, tikzit draw=blue]
\tikzstyle{blue}=[-, draw={blue!60!white}, tikzit draw=blue]
\tikzstyle{brace edge}=[-, tikzit draw=blue, decorate, decoration={brace,amplitude=1mm,raise=-1mm}]
\tikzstyle{diredge}=[->]

\makeatletter
\newcommand\etc{etc\@ifnextchar.{}{.\@}\xspace}
\makeatother

\newcommand{\intf}[1]{\left\llbracket #1 \right\rrbracket} % interpretation functor

\usepackage{bm}

\newcommand{\bra}[1]{\ensuremath{\left\langle #1 \right|}}
\newcommand{\ket}[1]{\ensuremath{\left|  #1 \right\rangle}}

\newcommand{\ketbra}[2]{\ensuremath{\ket{#1}\!\bra{#2}}}
\newcommand{\C}{\mathbb{C}}

\newcommand{\N}{\mathbb{N}}

\newcommand{\zh}{\text{ZH}}

\tikzstyle{dotpic}=[] % for backwards-compatibility
\tikzstyle{semilarge box}=[rectangle,inline text,fill=white,draw,minimum height=5mm,yshift=-0.5mm,minimum width=12.5mm,font=\small]
\tikzstyle{inline text}=[text height=1.5ex, text depth=0.25ex,yshift=0.5mm]
\tikzstyle{label}=[font=\footnotesize,text height=1.5ex, text depth=0.25ex]
\tikzstyle{white phase dot}=[white dot]

\newcommand{\dotunit}[1]{%
\,\begin{tikzpicture}[dotpic,yshift=1.5mm]
\node [#1] (a) at (0,-0.35) {}; 
\draw (a)--(0,0.3);
\end{tikzpicture}\,\xspace}

\newcommand{\dotmult}[1]{%
\,\begin{tikzpicture}[dotpic]
    \node [#1] (a) {};
    \draw (a) -- (90:0.55);
    \draw (a) (-45:0.6) -- (a);
    \draw (a) (-135:0.6) -- (a);
\end{tikzpicture}\,\xspace}

\newkeycommand{\phase}[style=white phase dot][1]{\,\begin{tikzpicture}
    \begin{pgfonlayer}{nodelayer}
        \node [style=none] (0) at (0, 0.6) {};
        \node [style=\commandkey{style}] (2) at (0, -0) {$#1$}; 
        \node [style=none] (3) at (0, -0.6) {};
    \end{pgfonlayer}
    \begin{pgfonlayer}{edgelayer}
        \draw (2) to (0.center);
        \draw (3.center) to (2);
    \end{pgfonlayer}
\end{tikzpicture}\,}

\newcommand{\gendiagram}[1]{
\begin{tikzpicture}
    \begin{pgfonlayer}{nodelayer}
        \node [style=none] (0) at (-0.75, 1) {};
        \node [style=none] (1) at (0.75, 1) {};
        \node [style=none] (2) at (-0.75, -1) {};
        \node [style=none] (3) at (0.75, -1) {};
        \node [style=none] (4) at (0, 0.75) {$\ldots$};
        \node [style=semilarge box] (5) at (0, -0) { #1 };
        \node [style=none] (6) at (0, -0.75) {$\ldots$};
    \end{pgfonlayer}
    \begin{pgfonlayer}{edgelayer}
        \draw (0.center) to (2.center);
        \draw (1.center) to (3.center);
    \end{pgfonlayer}
\end{tikzpicture}
}

\newcommand{\hadastate}[1]{\,\tikz{\node[style=hadamard] (x) {$#1$};\draw(x)--(0,0.75);}\,}

\newcommand{\hadaunit}{\dotunit{small hadamard}}

\newcommand{\hadamult}{\dotmult{small hadamard}}

\newcommand{\graymult}{\dotmult{gray dot}}
\newcommand{\grayphase}[1]{\phase[style=gray phase dot]{#1}}

\def\titlerunning{Completeness of the Phase-free ZH-calculus}
\title{\titlerunning}

\author{John van de Wetering
\institute{Radboud Universiteit}
\email{john@vdwetering.name} \and
    Sal Wolffs
\institute{Radboud Universiteit}
\email{sal.wolffs@gmail.com}
}
\def\authorrunning{J.~van de Wetering \& S.~Wolffs}

\begin{document}
\maketitle

\begin{abstract}
    The ZH-calculus is a graphical calculus for linear maps between qubits that allows a natural representation of the Toffoli+Hadamard gate set. The original version of the calculus, which allows every generator to be labelled by an arbitrary complex number, was shown to be complete by Backens and Kissinger. Even though the calculus is complete, this does not mean it allows one to easily reason in restricted settings such as is the case in quantum computing. In this paper we study the fragment of the ZH-calculus that is phase-free, and thus is closer aligned to physically implementable maps. We present a modified rule-set for the phase-free ZH-calculus and show that it is complete. We further discuss the minimality of this rule-set and we give an intuitive interpretation of the rules. Our completeness result follows by reducing to Vilmart's rule-set for the phase-free $\dzx$-calculus.
\end{abstract}

\section{Introduction}
Graphical languages are a powerful way to reason about quantum processes~\cite{CKbook}. The arbitrary connectivity of the generators make it perfectly suited for describing highly compositional tasks, such as measurement-based quantum computing~\cite{DP2}, lattice surgery~\cite{horsman2017surgery}, error-correcting codes~\cite{chancellor2016graphical} and circuit optimisation~\cite{FaganDuncan,cliff-simp,optimisation-paper}.

There are a variety of graphical languages in existence, each of which has different strengths and weaknesses. The most well-known is the ZX-calculus, which was introduced more than 10 years ago~\cite{CD1,CD2}. The generators are based on a pair of strongly complementary bases of the qubit, and this allows easy representation of stabilizer computation with arbitrary phase-gates. Another language, the ZW-calculus has its generators based on the two different types of entanglement that are possible in tripartite qubit systems. A variation of this language has found use in describing fermionic quantum computing~\cite{hadzihasanovic2018diagrammatic}.

The graphical calculus we study in this paper is the ZH-calculus, which was introduced last year by Backens and Kissinger~\cite{backens2018zhcalculus}. In this calculus, a new generator, the H-box, was introduced that generalised the Hadamard gate to arbitrary arity. Whereas the ZX-calculus allows easy representation of stabilizer and phase gates, the ZH-calculus allows one to easily represent the Hadamard+Toffoli gate set.

An important property for any graphical language is \emph{completeness}. Every diagram can be interpreted as a linear map between qubits. The language is complete when two diagrams can be rewritten into one another whenever they represent the same linear map. The completeness of the ZX-calculus was an open problem for several years~\cite{Witt:2014aa}, until finally a complete rule-set was found~\cite{SimonCompleteness,JPV-universal,ng2017universal}. The proof of completeness was based on the completeness of the ZW-calculus~\cite{HarnyAmarCompleteness,hadzihasanovic2017algebra}. The paper that introduced the ZH-calculus also immediately proved that the rule-set for this calculus was complete~\cite{backens2018zhcalculus}. 

The completeness proof of the ZH-calculus intensively used  the ability to label an H-box with an arbitrary complex number, which allows for an easy representation of a particular matrix normal-form. H-boxes labelled by most complex numbers unfortunately have no physical interpretation. So if we wish to use the ZH-calculus to reason, for instance, about quantum circuits, particularly those in the Hadamard+Toffoli gate set, then it would be useful to have a complete calculus of the restricted ZH-calculus where H-boxes are not labelled. We will refer to this restricted language as the \emph{phase-free ZH-calculus}.

In this paper we give a rule-set for the phase-free ZH-calculus that is complete. Our rule-set is made by taking all the rules of the full ZH-calculus and stripping away those rules that are trivial in the phase-free setting, and then adding two new rules (see Figure~\ref{fig:phasefree-rules}). The first is a simple rule to deal with zero scalars, but the second is more interesting. It relates two different ways of constructing the AND gate. These two constructions use H-boxes~\cite{backens2018zhcalculus}, respectively the \emph{triangle} of the ZX-calculus~\cite{vilmart2018zxtriangle}. This rule bridges the gap between the ZH-calculus and the ZX-calculus augmented with this triangle generator. As a result, we can use the completeness of Ref.~\cite{vilmart2018zxtriangle} to prove completeness of the phase-free ZH-calculus.

We also study the minimality of our rule set, and give an interpretation of most of the rules in terms of interactions of standard classical maps. Concerning this minimality, we conjecture that our new rules will actually prove not to be necessary. Hence, even though this new rule allows us to simplify quite a few proofs, we venture to prove as many results without using it.

The next section briefly recalls the definition of the ZH-calculus from Ref.~\cite{backens2018zhcalculus}. We then present the modified phase-free ZH-calculus in Section~\ref{sec:phase-free}, where we will also give a list of derived rewrite rules that we believe to be useful when reasoning about ZH-diagrams. The completeness proof is presented in Section~\ref{sec:completeness}, while the discussion on the minimality and interpretation of the rules happens in Section~\ref{sec:minimality}. Finally, a few concluding remarks are given in Section~\ref{sec:conclusion}

\section{The ZH-calculus}

The ZH-calculus is a graphical language expressing operations as \emph{string diagrams}.
These are diagrams consisting of dots or boxes, connected by wires.
Wires can also have one or two ends not connected to a dot or box. These represent inputs and outputs of the diagram when they exit towards the top respectively the bottom of the diagram.

Diagrams in the ZH-calculus have two types of generators, the \emph{Z-spiders} represented by white dots, and \emph{H-boxes} represented by white boxes. In the full ZH-calculus, the H-boxes are further labelled with a complex number.

These generators are interpreted as linear maps in the following manner:
\begin{equation*}
\intf{\input{Z-spider.tikz}} := \ket{0}^{\otimes n}\bra{0}^{\otimes m} + \ket{1}^{\otimes n}\bra{1}^{\otimes m} \qquad\qquad
 \intf{\input{H-spider.tikz}} := \sum a^{i_1\ldots i_m j_1\ldots j_n} \ket{j_1\ldots j_n}\bra{i_1\ldots i_m}
\end{equation*}

 where $\intf{\cdot}$ denotes the map from diagrams to matrices.

The sum in the second equation is over all $i_1,\ldots, i_m, j_1,\ldots, j_n\in\{0,1\}$ so that an H-box represents a matrix with all entries equal to 1, but for one element which is equal to $a$.
If the label of the H-box is~$-1$, then we do not write it, e.g.\ $\hadaunit:=\hadastate{\text{-}1}$.
%The H-box with label $z$ and no inputs or outputs represents the scalar $z$ and the Z-spider with no inputs or outputs represents the scalar 2.
Straight and curved wires have the following interpretations:
\begin{equation*}
\intf{\;|\;} := \ketbra{0}{0}+\ketbra{1}{1} \qquad\qquad\qquad
 \intf{\begin{tikzpicture}
	\begin{pgfonlayer}{nodelayer}
		\node [style=none] (0) at (-2, 0.25) {};
		\node [style=none] (1) at (-1.5, -0.25) {};
		\node [style=none] (2) at (-1, 0.25) {};
	\end{pgfonlayer}
	\begin{pgfonlayer}{edgelayer}
		\draw [bend right=45, looseness=1.00] (0.center) to (1.center);
		\draw [bend right=45, looseness=1.00] (1.center) to (2.center);
	\end{pgfonlayer}
\end{tikzpicture}} := \ket{00}+\ket{11} \qquad\qquad\qquad
 \intf{\input{wire-cap.tikz}} := \bra{00}+\bra{11}.
\end{equation*}

When two diagrams are juxtaposed, the corresponding linear map is the tensor product of the individual diagrams, while a sequential composition of two diagrams is interpreted as the matrix product of the corresponding matrices:
\[
 \intf{\gendiagram{$D_1$}\;\gendiagram{$D_2$}} := \intf{\gendiagram{$D_1$}}\otimes\intf{\gendiagram{$D_2$}} \qquad\qquad \intf{\input{sequential-composition.tikz}} := \intf{\gendiagram{$D_2$}}\circ\intf{\gendiagram{$D_1$}}
\]

To improve the presentation of the diagrams, we also use a few derived generators. The first two are \emph{grey spiders} and a spider with a \emph{NOT} on it:
\begin{equation*}
\input{X-spider-dfn.tikz} \qquad\qquad \input{negate-dfn.tikz}
\end{equation*}

The generator \graymult\ acts as XOR on the computational basis while \grayphase{\neg} acts as NOT:
\begin{equation*}
\intf{\graymult} = \ketbra{0}{00}+\ketbra{0}{11}+\ketbra{1}{01}+\ketbra{1}{10} \qquad\qquad\qquad \intf{\grayphase{\neg}}=\ketbra{0}{1}+\ketbra{1}{0}.
\end{equation*}

The ZH-calculus comes with a set of rewrite rules that was originally presented in Ref.~\cite{backens2018zhcalculus}. See Figure~\ref{fig:ZH-rules}. Additionally, the calculus has the meta-rule that \emph{only topology matters}. This means that two diagrams are considered equal when one can be topologically deformed into the other, while respecting the order of the inputs and outputs. Finally, the two generators are considered to be symmetric and undirected, so that the following equations also hold:
\ctikzfig{generator-symmetries}
These symmetry properties also hold for the derived grey spider and NOT gate.

\begin{figure}[!ht]
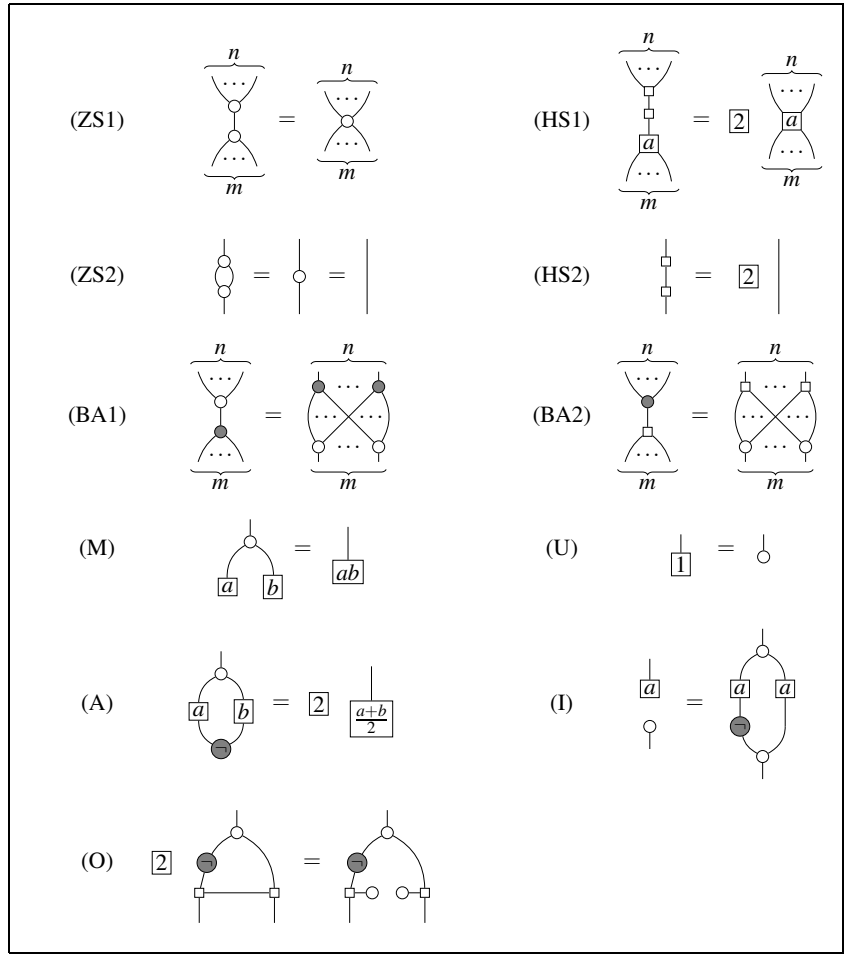

 \centering
 \scalebox{0.8}{%
 \begin{tabular}{|ccccc|}
 \hline 
 &&&&\\
  \qquad(ZS1) & \input{Z-spider-rule.tikz} & \qquad\qquad & (HS1) & \input{H-spider-rule-phased.tikz} \\ &&&& \\
  \qquad(ZS2) & \input{Z-special-phased.tikz} & & (HS2) & \input{H-identity-phased.tikz} \\ &&&& \\
  \qquad(BA1) & \input{ZX-bialgebra.tikz} & & (BA2) & \input{ZH-bialgebra.tikz} \qquad\\ &&&& \\
  \qquad(M) & \input{multiply-rule-phased.tikz} & & (U) & \input{unit-rule.tikz} \\ &&&& \\
  \qquad(A) & \input{average-rule.tikz} & & (I) & \input{intro-rule.tikz} \qquad\\ &&&& \\
  \qquad(O) & \input{ortho-rule-phased.tikz} & & & \\
  &&&&\\
  \hline
 \end{tabular}}
 \caption{The original set of rules for the ZH-calculus as presented in Ref.~\cite{backens2018zhcalculus}.
 Throughout, $m,n$ are nonnegative integers and $a,b$ are arbitrary complex numbers.
 The right-hand sides of both \textit{bialgebra} rules (BA1) and (BA2) are complete bipartite graphs on $(m+n)$ vertices, with an additional input or output for each vertex.
 The horizontal edges in equation (O) are well-defined because only the topology matters and we do not need to distinguish between inputs and outputs of generators. The rules (M), (A), (U), (I), and (O) are pronounced \textit{multiply}, \textit{average}, \textit{unit}, \textit{intro}, and \textit{ortho}, respectively.}
 \label{fig:ZH-rules}
\end{figure}

As shown in Ref.~\cite{backens2018zhcalculus}, the rules of Figure~\ref{fig:ZH-rules} and the meta-rules are \emph{sound} for the interpretation functor $\intf{\cdot}$, meaning that if two diagrams are equal with respect to these rules, then the linear maps they represent are also equal. Furthermore, ZH-diagrams are \emph{universal} for linear maps between qubits, meaning that for any linear map $f:(\C^2)^{\otimes n} \rightarrow (\C^2)^{\otimes m}$ we can find a ZH-diagram $D$ such that $\intf{D} = f$.

\section{The phase-free ZH-calculus}\label{sec:phase-free}

In this paper we will consider ZH-diagrams where the label of an H-box can only be $-1$. We will refer to ZH-diagrams like that as \emph{phase-free}. In order to state the definition of the grey spider and the NOT, which contain a label of $\frac12$, we introduce the generator \emph{star}, which has zero inputs and outputs. It's interpretation is:
\begin{equation*}
    \intf{\begin{tikzpicture}
	\begin{pgfonlayer}{nodelayer}
		\node [style=st] (3) at (0, 0) {};
	\end{pgfonlayer}
\end{tikzpicture}
} := \frac{1}{\sqrt{2}}
\end{equation*}

We can then define the grey spider and the spider with a NOT as:
\begin{equation*}\label{eq:defx}
(X) \quad\ \  \input{X-spider-dfn-free.tikz}\qquad\qquad\qquad (NOT)\quad\ \  \input{negate-dfn-free.tikz}
\end{equation*}

Going beyond the definition of the original paper introducing the ZH-calculus, we also introduce the derived \emph{negate} generator (a white NOT), and the \emph{triangle}:

\begin{equation}\label{eq:Z-triangle-dfn}
 (Z) \quad\ \ \input{negate-white-dfn.tikz}\qquad \qquad \qquad (\Delta)\quad\ \ \input{triangle-dfn.tikz}
\end{equation}
The negate white spider with one input and output acts like the Z gate, while the triangle is the same map as in the ZX-calculus~\cite{vilmart2018zxtriangle}:
\begin{equation*}
\intf{\phase{\neg}} = \ketbra{0}{0} - \ketbra{1}{1} \qquad\qquad\qquad \intf{\begin{tikzpicture}
	\begin{pgfonlayer}{nodelayer}
		\node [style=triangle] (8) at (0, 0) {};
		\node [style=none] (9) at (0, 1) {};
		\node [style=none] (10) at (0, -1) {};
	\end{pgfonlayer}
	\begin{pgfonlayer}{edgelayer}
		\draw (10.center) to (9.center);
	\end{pgfonlayer}
\end{tikzpicture}
} = \ketbra{0}{0}+\ketbra{1}{0} + \ketbra{0}{1}
\end{equation*}

As the triangle is not symmetric, we define its transpose as:
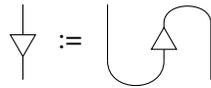
\ctikzfig{triangle-flip-dfn}

Let us consider how the rules of the full calculus shown in Figure~\ref{fig:ZH-rules} can be lifted to the phase-free calculus. Of the original rules, (HS1), (M), (A) and (I) allow arbitrary complex labels to appear. We make them phase-free, by simply restricting these rules to only apply to a label of $-1$. As it turns out, in that setting (A) and (I) are provable using the other rules, and hence we will not require them.
The rules (HS1), (HS2) and (O) include a label of 2. To make them phase-free we use a single Z-spider with zero legs to represent a scalar of 2. The rule (U) is unnecessary in the phase-free fragment.

To show completeness of the phase-free calculus we need to add two new rules. One rule simply states that the star generator can be removed when multiplied with the zero scalar.  The second rule is more interesting. It tells us that two different ways of representing an AND-gate, using either H-boxes or triangles, are the same thing:
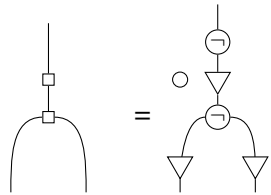
\ctikzfig{and-rule}
That the left-hand side of this equation is equal to an AND gate was established in the original ZH-calculus paper and can quite easily be verified by direct calculation. The right-hand side represents the original way in which the AND gate was constructed in the ZX-calculus~\cite[Exercise 12.10]{CKbook}.
Please see Figure~\ref{fig:phasefree-rules} for an overview of the rules of the phase-free fragment.

\begin{figure}[!bt]
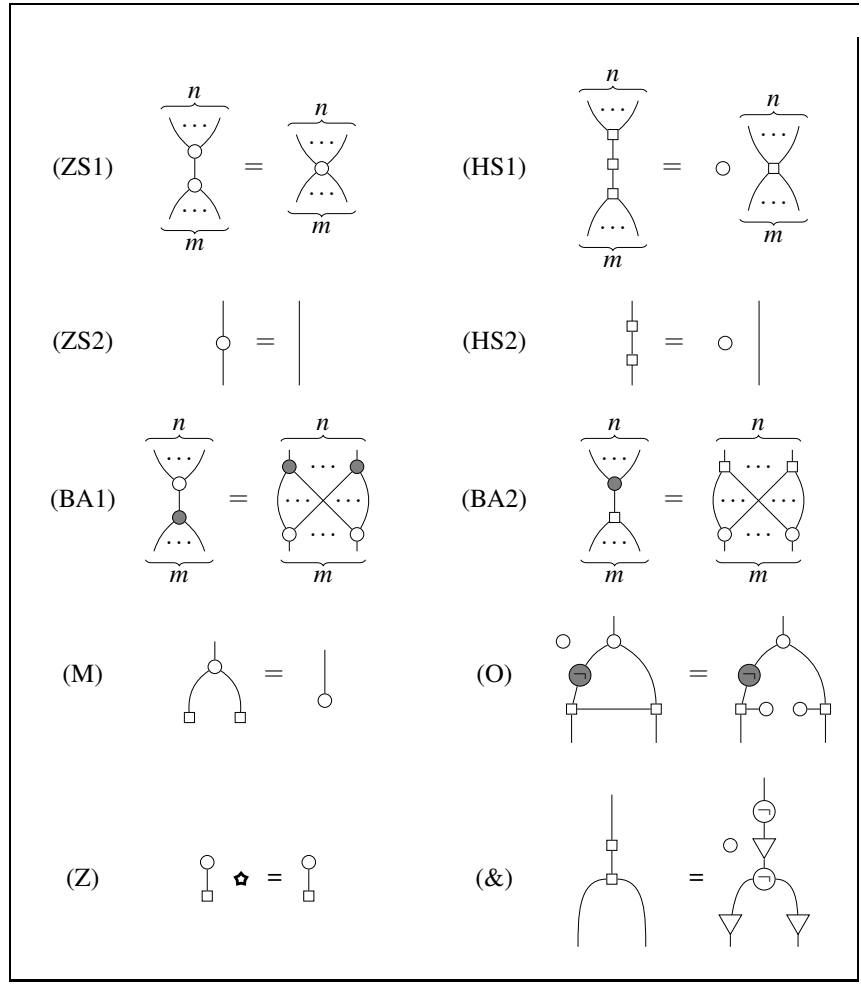

 \centering
 \scalebox{0.9}{%
 \begin{tabular}{|ccccc|}
 \hline \\
 &&&&\\
  \quad (ZS1) & \input{Z-spider-rule.tikz} &\quad \qquad & (HS1) & \input{H-spider-rule.tikz} \quad \\ &&&& \\
  \quad(ZS2) & \input{Z-special.tikz} & & (HS2) & \input{H-identity.tikz} \quad\\ &&&& \\
  \quad(BA1) & \input{ZX-bialgebra.tikz} & & (BA2) & \input{ZH-bialgebra.tikz} \quad\\ &&&& \\
  \quad(M) & \input{multiply-rule.tikz} & & (O) & \input{ortho-rule.tikz} \quad\\&&&& \\
  \quad(Z) & \input{star-zero-rule.tikz} & & (\&) & \input{and-rule.tikz} \quad\\
  &&&&\\
 \hline
 \end{tabular}}
 \caption{The rules for the phase-free ZH-calculus.
 Throughout, $m,n$ are nonnegative integers.
 The right-hand sides of both \textit{bialgebra} rules (BA1) and (BA2) are complete bipartite graphs on $(m+n)$ vertices, with an additional input or output for each vertex.
 The horizontal edges in equation (O) are well-defined because only the topology matters and we do not need to distinguish between inputs and outputs of generators. n.b. the rules (M), (O), (Z) and (\&) are pronounced \textit{multiply}, \textit{ortho}, \textit{zero} and \textit{and} respectively.}
 \label{fig:phasefree-rules}
\end{figure}

\subsection{Derived rules}

The original paper on the ZH-calculus only proved those results that were necessary to prove completeness. As a result, only a small amount of useful derived rewrite rules were presented in that paper. In this section, we aim to remedy this by providing a list of derived rewrite rules, that we believe to be useful for reasoning in the ZH-calculus. The proofs for these derivations can be found in Appendix~\ref{sec:derived-rules}. We label these derived rules with letters in order to more easily refer to them when used in the Appendix.

\begin{lemma}\label{lem:scalarcancel}Scalars can be introduced and cancelled in the following ways:
    \begin{equation*}
        (S) \qquad \input{scalar-rule.tikz} \qquad\qquad 
    \end{equation*}
\end{lemma}

\begin{lemma}\label{lem:zx-inspired} The following ZX-like rules hold:
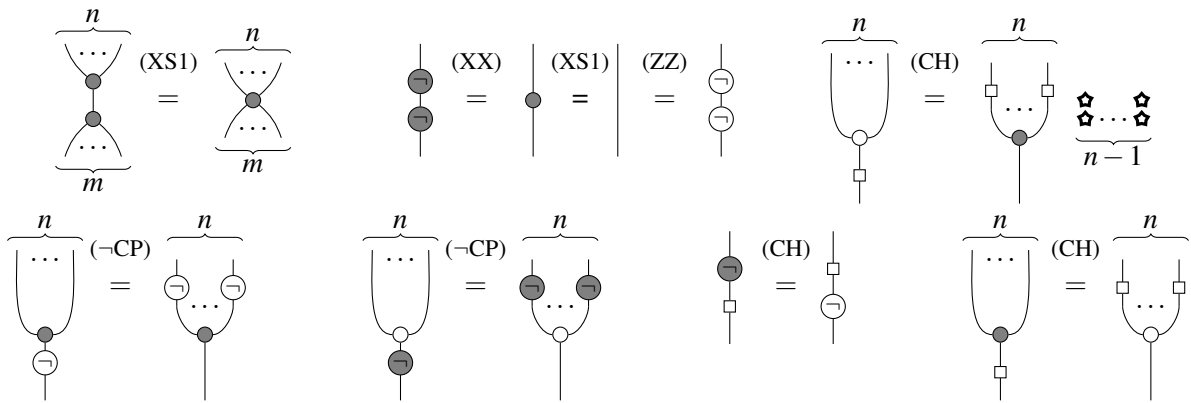
\begin{gather*}
  \qquad \input{X-spider-rule.tikz} \qquad  \qquad \input{XX-ZZ-cancel.tikz} \quad  \qquad \input{H-Z-commute.tikz} \qquad\\
\input{Z-commute.tikz}  \qquad  \qquad \input{NOT-commute.tikz} \qquad \qquad \input{H-NOT-commute.tikz}  \qquad \qquad \input{H-X-commute.tikz}
\end{gather*}
\end{lemma}

\begin{lemma}\label{lem:statecopy} The following rules for copying states hold:
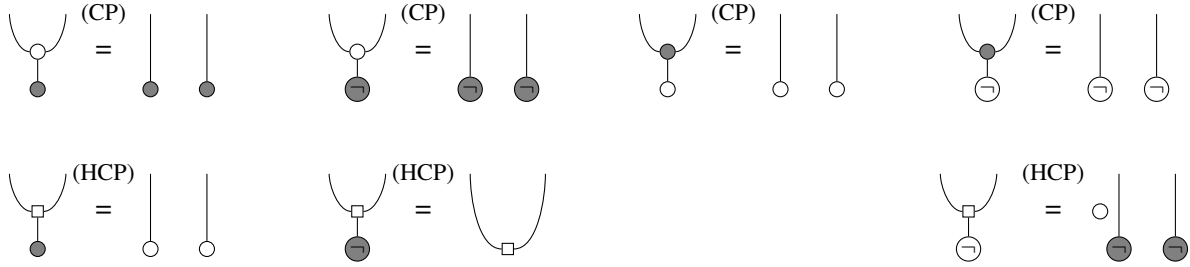
\ctikzfig{state-copy}
\end{lemma}

\begin{lemma}\label{lem:commutation} The following commutation and simplifying rules hold:
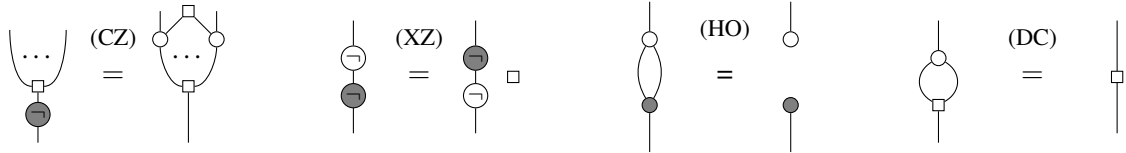
\begin{equation*}
  \input{CZ-correct.tikz} \qquad \qquad \input{X-Z-commute.tikz} \qquad \qquad \input{hopf-rule.tikz} \qquad \qquad \input{had-Z-cancel.tikz}
\end{equation*}
\end{lemma}

All the derived rules presented so far follow relatively easy from the rules. The exceptions are \HadDoubleRule and the rightmost \StateCopyH of Lemma~\ref{lem:statecopy}. For \StateCopyH our proof requires \OrthoRule, while for \HadDoubleRule we only know how to prove it by using \AndRule or by using a lengthy proof inspired by the normal-form of Ref.~\cite{backens2018zhcalculus} (see Appendix~\ref{sec:proof-of-dedup}).

\begin{lemma}\label{lem:triangle-rules} The following rules involving the triangle hold:
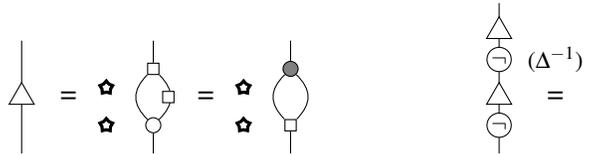
\begin{equation*}
    \input{triangle-def2.tikz} \qquad \qquad \qquad \input{triangle-inverse.tikz}
\end{equation*}
\end{lemma}

As we have not needed \AndRule to prove all the results up to now, we can use \TriangleInverse to state an equivalent form of \AndRule:

\begin{lemma}\label{lem:and-rule-prime}
    In the presence of all the other rules, \AndRule is equivalent to:
    \begin{equation}\label{eq:and-rule-prime}
     \input{and-rule-prime.tikz}
    \end{equation}
\end{lemma}

\section{Completeness}\label{sec:completeness}

We write $\zh \vdash D_1 = D_2$ when $D_1$ and $D_2$ can be proven to be equal using the rules of the phase-free ZH-calculus. The ZH-calculus is \emph{complete} when $\intf{D_1}=\intf{D_2}$ implies that ${\zh \vdash D_1=D_2}$.

To prove completeness we use the result that there is a rule-set for the $\pi$-fragment of the ZX-calculus with an additional triangle node that is complete~\cite{vilmart2018zxtriangle}. We will refer to this calculus as \dzx.
We will give interpretations going back and forth between phase-free ZH-diagrams and \dzx and show that the rules of the ZH-calculus suffice to prove all the rules of \dzx under this interpretation.

\subsection{The \texorpdfstring{\dzx}{DZX}-calculus}

The \dzx-calculus has four generators: Z- and X-spiders, Hadamard gates and triangles. The Hadamard-gate is depicted as a yellow box and the triangle as a red triangle:
\begin{equation*}
\intf{~\begin{tikzpicture}
	\begin{pgfonlayer}{nodelayer}
		\node [style=H box] (0) at (0, 0) {};
		\node [style=none] (1) at (0, 1) {};
		\node [style=none] (2) at (0, -1) {};
	\end{pgfonlayer}
	\begin{pgfonlayer}{edgelayer}
		\draw (2.center) to (1.center);
	\end{pgfonlayer}
\end{tikzpicture}
~}:= \frac{1}{\sqrt{2}}\begin{pmatrix}1 & 1\\1 & -1\end{pmatrix}\qquad\qquad
\intf{\begin{tikzpicture}
	\begin{pgfonlayer}{nodelayer}
		\node [style=none] (0) at (0, 1) {};
		\node [style=none] (1) at (0, -1) {};
		\node [style=ug] (2) at (0, 0) {};
	\end{pgfonlayer}
	\begin{pgfonlayer}{edgelayer}
		\draw [style=none] (0.center) to (1.center);
	\end{pgfonlayer}
\end{tikzpicture}
}:=\begin{pmatrix}1&1\\0&1\end{pmatrix}
\end{equation*}
The Z- and X-spiders are depicted respectively as green and red dots with an internal phase:
\begin{equation*}
\intf{\input{gn-alpha.tikz}} = \ket{0}^{\otimes n}\bra{0}^{\otimes m} + e^{i\alpha}\ket{1}^{\otimes n}\bra{1}^{\otimes m} \qquad \qquad \intf{\input{rn-alpha.tikz}}:=\intf{~\begin{tikzpicture}
	\begin{pgfonlayer}{nodelayer}
		\node [style=H box] (0) at (0, 0) {};
		\node [style=none] (1) at (0, 1) {};
		\node [style=none] (2) at (0, -1) {};
	\end{pgfonlayer}
	\begin{pgfonlayer}{edgelayer}
		\draw (2.center) to (1.center);
	\end{pgfonlayer}
\end{tikzpicture}
~}^{\otimes m}\circ \intf{\input{gn-alpha.tikz}}\circ \intf{~\begin{tikzpicture}
	\begin{pgfonlayer}{nodelayer}
		\node [style=H box] (0) at (0, 0) {};
		\node [style=none] (1) at (0, 1) {};
		\node [style=none] (2) at (0, -1) {};
	\end{pgfonlayer}
	\begin{pgfonlayer}{edgelayer}
		\draw (2.center) to (1.center);
	\end{pgfonlayer}
\end{tikzpicture}
~}^{\otimes n}
\end{equation*}
For our purposes the value of $\alpha$ is either $0$, in which case we do not write it, or $\pi$.

The \dzx-calculus has the same meta-rules as the ZH-calculus: only connectivity matters, and the generators are symmetric with respect to swaps and transposes, with the exception of the triangle, which when transposed gets flipped upside down in the same way as the triangle in the ZH-calculus. The other rules of the \dzx-calculus are presented in Figure~\ref{fig:triangle-rules}.

\begin{figure*}[!bt]
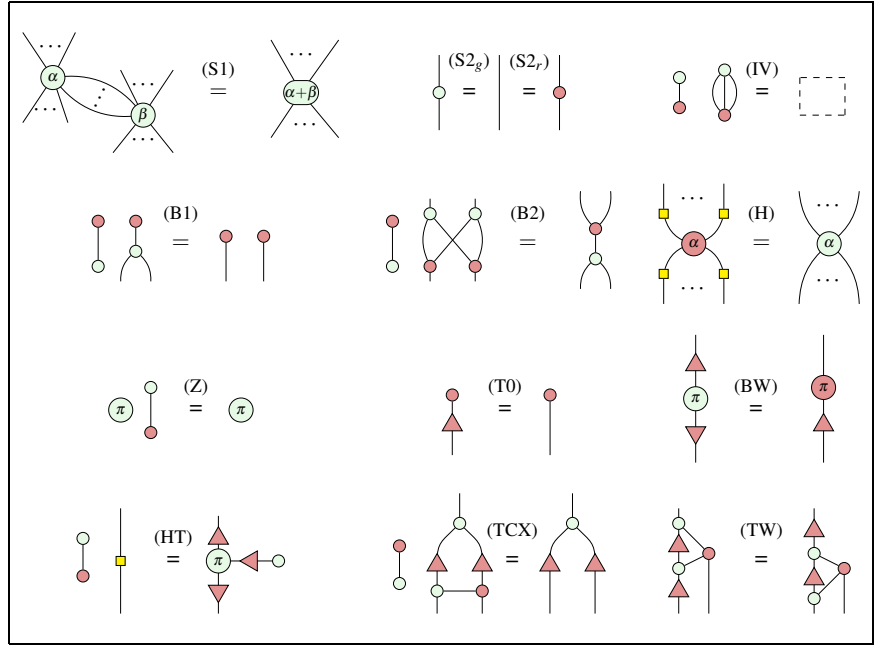

 \centering
 \scalebox{0.8}{%
 \begin{tabular}{|c@{$\qquad$}c@{$\qquad$}c|}
 \hline
   && \\
   \input{spider-1.tikz}&\input{s2-green-red.tikz}&\input{inverse.tikz}\\
   && \\
   \input{b1s.tikz}& \input{b2s.tikz}&\input{h2.tikz}\\
   && \\
   \input{zero-rule.tikz}&\input{ket-0-on-triangle.tikz}&\input{BW.tikz}\\
   && \\
   \input{triangle-hadamard.tikz}&\input{cnot-on-triangle-fork.tikz}&\input{triangle-through-W.tikz}\\
   && \\
   \hline
  \end{tabular}}
 \caption{The complete set of rules of Ref.~\cite{vilmart2018zxtriangle} for the ZX-Calculus with triangles. The right-hand side of (IV) is an empty diagram. (...) denote zero or more wires, while (\protect\rotatebox{45}{\raisebox{-0.4em}{$\cdots$}}) denote one or more wires. $\alpha,\beta\in\{0,\pi\}$.
 }
 \label{fig:triangle-rules}
\end{figure*}

\subsection{From \texorpdfstring{\dzx}{DZX} to ZH}

In this section we give an interpretation from \dzx-diagrams to ZH-diagrams that preserves the linear map semantics. We will denote this interpretation by $\intf{\cdot}_{\text{ZH}}$. To do this, we must give an interpretation of every generator of the \dzx-calculus as a ZH-diagram. This is very straightforward, only requiring a normalisation for the Hadamard gate:
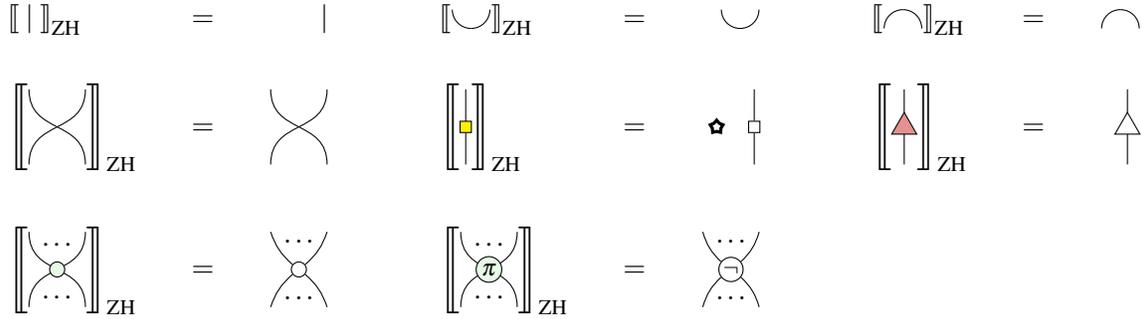
\begin{equation*}\def\arraystretch{3.0}
\begin{array}{lcrclcrclcr}
\intf{\;|\;}_\zh & \quad=\quad & | &\qquad &\intf{\begin{tikzpicture}
	\begin{pgfonlayer}{nodelayer}
		\node [style=none] (0) at (-2, 0.25) {};
		\node [style=none] (1) at (-1.5, -0.25) {};
		\node [style=none] (2) at (-1, 0.25) {};
	\end{pgfonlayer}
	\begin{pgfonlayer}{edgelayer}
		\draw [bend right=45, looseness=1.00] (0.center) to (1.center);
		\draw [bend right=45, looseness=1.00] (1.center) to (2.center);
	\end{pgfonlayer}
\end{tikzpicture}}_\zh & \quad=\quad & \begin{tikzpicture}
	\begin{pgfonlayer}{nodelayer}
		\node [style=none] (0) at (-2, 0.25) {};
		\node [style=none] (1) at (-1.5, -0.25) {};
		\node [style=none] (2) at (-1, 0.25) {};
	\end{pgfonlayer}
	\begin{pgfonlayer}{edgelayer}
		\draw [bend right=45, looseness=1.00] (0.center) to (1.center);
		\draw [bend right=45, looseness=1.00] (1.center) to (2.center);
	\end{pgfonlayer}
\end{tikzpicture} & \qquad & \intf{\input{wire-cap.tikz}}_\zh  & \quad=\quad & \input{wire-cap.tikz} \\ \vspace{0.4cm}
\intf{\input{swap.tikz}}_\zh &\quad=\quad& \input{swap.tikz} & \qquad & \intf{\begin{tikzpicture}
	\begin{pgfonlayer}{nodelayer}
		\node [style=H box] (0) at (0, 0) {};
		\node [style=none] (1) at (0, 1) {};
		\node [style=none] (2) at (0, -1) {};
	\end{pgfonlayer}
	\begin{pgfonlayer}{edgelayer}
		\draw (2.center) to (1.center);
	\end{pgfonlayer}
\end{tikzpicture}
}_\zh &\quad=\quad& \begin{tikzpicture}[decoration=brace]
	\begin{pgfonlayer}{nodelayer}
		\node [style=hadamard] (0) at (0.5, 0) {};
		\node [style=none] (1) at (0.5, 1) {};
		\node [style=none] (2) at (0.5, -1) {};
		\node [style=st] (3) at (-0.5, 0) {};
	\end{pgfonlayer}
	\begin{pgfonlayer}{edgelayer}
		\draw (2.center) to (0);
		\draw (0) to (1.center);
	\end{pgfonlayer}
\end{tikzpicture}
 & \qquad & \intf{\begin{tikzpicture}
	\begin{pgfonlayer}{nodelayer}
		\node [style=none] (0) at (0, 1) {};
		\node [style=none] (1) at (0, -1) {};
		\node [style=ug] (2) at (0, 0) {};
	\end{pgfonlayer}
	\begin{pgfonlayer}{edgelayer}
		\draw [style=none] (0.center) to (1.center);
	\end{pgfonlayer}
\end{tikzpicture}
}_\zh & \quad=\quad & \begin{tikzpicture}
	\begin{pgfonlayer}{nodelayer}
		\node [style=none] (0) at (0, 1) {};
		\node [style=none] (1) at (0, -1) {};
		\node [style=triangle] (2) at (0, 0) {};
	\end{pgfonlayer}
	\begin{pgfonlayer}{edgelayer}
		\draw [style=none] (0.center) to (1.center);
	\end{pgfonlayer}
\end{tikzpicture}
 \\ \vspace{0.4cm}
\intf{\input{gn.tikz}}_\zh & \quad = \quad & \input{Z-spider-no-label.tikz} & \qquad & \intf{\input{gn-pi.tikz}}_\zh & \quad = \quad & \input{Z-spider-negate.tikz}
\end{array}
\end{equation*}

\noindent The interpretation of the X-spider follows from the interpretation of the Z-spider and the Hadamard gate. It is straightforward to check that this interpretation preserves the semantics:
\begin{proposition}\label{prop:preservesemantics}
Let $D$ be a \dzx-diagram. Then $\intf{D} = \intf{\intf{D}_\zh}$.
\end{proposition}
 
\noindent As a result we get the following proposition:

\begin{proposition}\label{prop:universality}
  The phase-free ZH-calculus is \emph{universal} for the matrices over $\mathbb{Z}[\frac{1}{\sqrt{2}}]$ of shape $2^n\times 2^m$. I.e.~let $M$ be any $2^n\times 2^m$ matrix for $n,m \in \N$, with entries of the form $\frac{a}{\sqrt{2}^b}$ where $a\in\mathbb{Z}$ and $b\in\N$. Then there exists a phase-free ZH-diagram $D$ such that $\intf{D} = M$.
\end{proposition}
\begin{proof}
  As shown in Ref.~\cite{vilmart2018zxtriangle}, the \dzx-calculus is universal for this set of matrices, hence by proposition~\ref{prop:preservesemantics}, the same holds true for ZH-diagrams.
\end{proof}

We write $\dzx \vdash D_1=D_2$ when $D_1$ and $D_2$ are \dzx-diagrams that can be proven to be equal using the \dzx-calculus. The following proposition lies at the heart of our completeness result. It essentially shows that the ZH-calculus can prove all the equations that the \dzx-calculus can.

\begin{proposition}\label{prop:dzx-to-zh}
  Let $D_1$ and $D_2$ be $\dzx$-diagrams such that $\dzx \vdash D_1=D_2$. Then  $\zh \vdash \intf{D_1}_\zh = \intf{D_2}_\zh$.
\end{proposition}
\begin{proof}
  The proof is given in Appendix~\ref{sec:proof-of-dzx-to-zh}.
\end{proof}

\subsection{From ZH to \texorpdfstring{\dzx}{DZX}}

We now give an interpretation of phase-free ZH-diagrams into \dzx-diagrams. This is again straightforward, with the only complication coming with the interpretation of the H-boxes. Since $n$-arity H-boxes for $n>3$ can be seen as derived generators by axiom \HFuseRule, we only give an interpretation of the H-box for arity $n=1,2,3$. Note that even though \HFuseRule does not give a unique way to decompose a higher arity H-box into a lower arity one, because of the completeness of the \dzx-calculus, any translation of a decomposition into the \dzx-calculus would lead to a provably equal \dzx-diagram, and hence we can indeed define our interpretation inductively like this. We denote the interpretation by $\intf{\cdot}_\dzx$. It is given by:

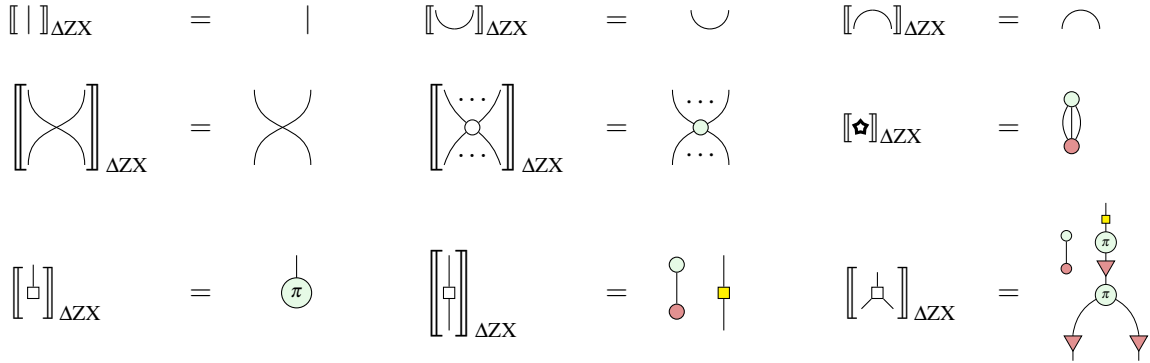
\begin{equation*}\def\arraystretch{3.0}
\begin{array}{lcrclcrclcl}
\intf{\;|\;}_\dzx & \ \ =\ \  & | &\qquad &\intf{\begin{tikzpicture}
	\begin{pgfonlayer}{nodelayer}
		\node [style=none] (0) at (-2, 0.25) {};
		\node [style=none] (1) at (-1.5, -0.25) {};
		\node [style=none] (2) at (-1, 0.25) {};
	\end{pgfonlayer}
	\begin{pgfonlayer}{edgelayer}
		\draw [bend right=45, looseness=1.00] (0.center) to (1.center);
		\draw [bend right=45, looseness=1.00] (1.center) to (2.center);
	\end{pgfonlayer}
\end{tikzpicture}}_\dzx & \ \ =\ \  & \begin{tikzpicture}
	\begin{pgfonlayer}{nodelayer}
		\node [style=none] (0) at (-2, 0.25) {};
		\node [style=none] (1) at (-1.5, -0.25) {};
		\node [style=none] (2) at (-1, 0.25) {};
	\end{pgfonlayer}
	\begin{pgfonlayer}{edgelayer}
		\draw [bend right=45, looseness=1.00] (0.center) to (1.center);
		\draw [bend right=45, looseness=1.00] (1.center) to (2.center);
	\end{pgfonlayer}
\end{tikzpicture} & \qquad & \intf{\input{wire-cap.tikz}}_\dzx  & \ \ =\ \  & \input{wire-cap.tikz} \\ \vspace{0.4cm}

\intf{\input{swap.tikz}}_\dzx &\ \ =\ \ & \input{swap.tikz} & \qquad & \intf{\input{Z-spider-no-label.tikz}}_\dzx &\ \ =\ \ & \input{gn.tikz}  & \qquad & \intf{\begin{tikzpicture}
	\begin{pgfonlayer}{nodelayer}
		\node [style=st] (3) at (0, 0) {};
	\end{pgfonlayer}
\end{tikzpicture}
}_\dzx & \ \ =\ \ &  \begin{tikzpicture}[decoration=brace]
	\begin{pgfonlayer}{nodelayer}
		\node [style=Z dot] (3) at (0, 0.75) {};
		\node [style=X dot] (4) at (0, -0.5) {};
	\end{pgfonlayer}
	\begin{pgfonlayer}{edgelayer}
		\draw (3) to (4);
		\draw [bend right] (3) to (4);
		\draw [bend right] (4) to (3);
	\end{pgfonlayer}
\end{tikzpicture}
\\ \vspace{0.4cm}

\intf{\hadastate{}}_\dzx &\ \ =\ \ & \begin{tikzpicture}
	\begin{pgfonlayer}{nodelayer}
		\node [style=Z phase dot] (0) at (0, 0) {$\pi$};
		\node [style=none] (1) at (0, 1) {};
	\end{pgfonlayer}
	\begin{pgfonlayer}{edgelayer}
		\draw (1.center) to (0);
	\end{pgfonlayer}
\end{tikzpicture}
 & \qquad & \intf{\begin{tikzpicture}[decoration=brace]
	\begin{pgfonlayer}{nodelayer}
		\node [style=hadamard] (0) at (0, 0) {};
		\node [style=none] (1) at (0, 1) {};
		\node [style=none] (2) at (0, -1) {};
	\end{pgfonlayer}
	\begin{pgfonlayer}{edgelayer}
		\draw (2.center) to (0);
		\draw (0) to (1.center);
	\end{pgfonlayer}
\end{tikzpicture}
}_\dzx & \ \ =\ \  & \input{Hadamard-normalised.tikz} & \qquad &
\intf{\hadamult}_\dzx & \ \  = \ \  & \scalebox{0.7}{\input{AND-triangle.tikz}} 
\end{array}
\end{equation*}

\noindent It is again straightforward to check that this interpretation preserves semantics:
\begin{proposition}\label{prop:semantics2}
  Let $D$ be a ZH-diagram. Then $\intf{D} = \intf{\intf{D}_\dzx}$.
\end{proposition}

More interestingly, the ZH-calculus can prove that the composition of the interpretations, going into the \dzx-calculus, and back to the ZH-calculus gives the same linear maps:

\begin{proposition}\label{prop:zh-dzx-loop}
  Let $D$ be a ZH-diagram. $\zh \vdash \intf{\intf{D}_\dzx}_\zh = D$.
\end{proposition}
\begin{proof}
  For the details see Appendix~\ref{sec:proof-of-zh-dzx-loop}. The only non-trivial step is to check that the result holds with $D$ equal to the 3-ary H-box, for which rule \AndRule is used.
\end{proof}

\noindent We can now prove that the phase-free ZH-calculus is indeed complete:
\begin{theorem}
  The phase-free ZH-calculus is complete.
\end{theorem}
\begin{proof}
  Let $D_1$ and $D_2$ be ZH-diagrams with $\intf{D_1}=\intf{D_2}$. By Proposition~\ref{prop:semantics2} we then have $\intf{\intf{D_1}_\dzx} = \intf{\intf{D_2}_\dzx}$. By completeness of the \dzx-calculus~\cite{vilmart2018zxtriangle} we then have $\dzx \vdash \intf{D_1}_\dzx = \intf{D_2}_\dzx$. Proposition~\ref{prop:dzx-to-zh} then gives us $\zh \vdash \intf{\intf{D_1}_\dzx}_\zh = \intf{\intf{D_2}_\dzx}_\zh$. Finally, by transitivity of equality and Proposition~\ref{prop:zh-dzx-loop} we conclude that $\zh \vdash D_1=D_2$.
\end{proof}

\section{Minimality and interpretation of the rules}\label{sec:minimality}

In this section we will give an interpretation of the rules as stating identities concerning the classical COPY, XOR, NOT and AND gates, and we will give a few results regarding minimality of the rules.

We start by studying the minimality of the rules. The following rules (or combinations of rules) are necessary, because they are the only ones to break certain invariants of the diagrams:
\begin{itemize}
  \item Both \SpiderRule and \HFuseRule are needed since they are the only rules that relate arity 4 spiders, respectively H-boxes, to lower arity ones.
  \item \ZeroRule is needed, since this is the only rule that changes the parity of the amount of star generators in the diagram.
  \item At least one of \IDRule or \HHRule is needed, since these are the only rules that relate an empty wire to a non-empty wire.
  \item At least one of \StrongCompRule and \HCompRule with $n=m=0$ is needed, since those are the only rules that relate a non-empty diagram to an empty one.
  %\item At least one of \IDRule or \HCompRule is needed, since these are the only rules that can change the amount of Z-spiders in a diagram from a non-zero amount to a zero amount.
  %\item At least one of \HHRule or \NotRule is needed, since these are the only rules that can change the amount of H-boxes in a diagram from a non-zero amount to a zero amount.
\end{itemize}

For showing the necessity of the following pairs of rules, we construct different interpretations. In these interpretations, all the rules, except the ones under investigation, are true. Hence, those rules must be independent from the others. We will consider the rule-set of Figure~\ref{fig:phasefree-rules}, but with \AndRule replaced by \eqref{eq:and-rule-prime}, which is an equivalent rule by Lemma~\ref{lem:and-rule-prime}. For the interpretations we consider equality up to non-zero scalar, and we set $\begin{tikzpicture}
	\begin{pgfonlayer}{nodelayer}
		\node [style=st] (3) at (0, 0) {};
	\end{pgfonlayer}
\end{tikzpicture}
=\input{empty-diagram.tikz}$, so that \ZeroRule becomes trivial.
\begin{itemize}
  \item At least one of \HHRule or \StrongCompRule is needed, because the interpretation where we change the H-box so that\ \  $\input{H-box-as-disconnected.tikz}~$ satisfies all other rules.
  \item At least one of \IDRule or \HCompRule is needed, because the interpretation where we change the Z-spider so that \ \ $\input{Z-spider-as-disconnected.tikz}~$ satisfies all other rules.
  \item At least one of \HCompRule or \OrthoRule is needed, because the interpretation where we change the H-spider so that\ \  $\input{H-box-as-white-dots.tikz}~$ satisfies all other rules.
\end{itemize}

Of course the full generality of rules \StrongCompRule and \HCompRule are not needed. For \StrongCompRule $n=m=0$, $n=2,m=0$ and $n=m=2$ is sufficient, as the remaining rules follow by induction (see e.g.~\cite{CKbook}), while for \HCompRule we also need $n=0,m=2$ due to a lack of colour symmetry. With this taken into account, we conjecture that all the rules except for \AndRule and \OrthoRule are needed, with at least one of these also being necessary.

\subsection{Interpreting the rules}
Most of the rules state a property of certain classical maps, as we will now show. We make the following identifications:

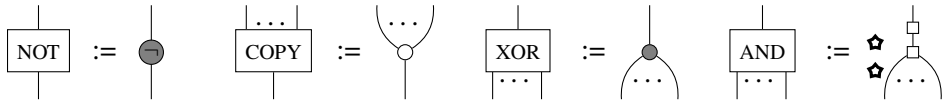
\ctikzfig{classical-interpretation}

The spider rules then simply state the following evidently true identities for these classical maps:

\ctikzfig{classical-spider}

Following \cite[Proposition 8.19]{CKbook}, the bialgebra rules \StrongCompRule and \HCompRule state that XOR and AND are \emph{function maps}, i.e.\ that they send computational basis states to computational basis states:

\ctikzfig{classical-bialgebra}

\noindent In the presence of the other rules it is straightforward to verify that \NotRule is equivalent to two NOT gates cancelling:
\begin{equation*}
    \input{multiply-rule.tikz} \qquad \iff \qquad \input{classical-notnot.tikz} 
\end{equation*}

Finally, \OrthoRule can also be seen as a statement about the classical functions by changing the output into an input:
\begin{equation*}
    \input{ortho-rule.tikz} \quad \iff \quad \input{ortho-rule-mod.tikz} \quad \iff \quad \input{classical-ortho.tikz}
\end{equation*}
Written in this way we clearly see why \OrthoRule is true: the copied input is sent to both AND's, but for one it is negated. Hence the output of at least one of the AND gates will be zero so that the only way for the outputs of these AND gates to be equal, is if they are both zero.

The only rules we can not give any classical interpretation for are the newly added ones \ZeroRule and \AndRule, but considering that most of our proofs do not use these rules it raises the following question: how is it that we can prove things about quantum maps using these rules that seem to say something about classical maps? The missing pieces here are that the generators can be combined in more interesting ways due to the presence of cups and caps, and that we have a map, the Hadamard, that does two things at the same time: it converts COPY's into XOR's and back, and it can be composed with an AND gate in order to make it symmetric with respect to yanking wires. Somehow the existence of a map like that is sufficient to recover all the structure of linear maps between qubits.

\section{Conclusion and further work}\label{sec:conclusion}

We have given a rule-set for the phase-free ZH-calculus that is complete. Apart from a simple rule to deal with zero scalars, we have added one more new rule which relates the H-box to the triangle node used in the ZX-calculus. We have discussed the minimality of the used rules, and have given a classical interpretation of most of the rules. 

The completeness proof uses a reduction to the \dzx-calculus of Ref.~\cite{vilmart2018zxtriangle}, which in turn uses the completeness of the ZW-calculus~\cite{HarnyAmarCompleteness}. It would be interesting to see if a more direct proof of completeness can be constructed, i.e.\ by mimicking the way the normal-form is constructed in Ref.~\cite{HarnyAmarCompleteness}. In particular, because our proof used this reduction, we needed to add a generator to represent the scalar $\frac{1}{\sqrt{2}}$, which requires us to include the rule \ZeroRule. If the completeness can be shown directly, then we might be able to replace this by a generator to represent $\frac12$, which does not need such a rule.

It is not clear whether the rule \AndRule is necessary for completeness. We only use it to prove the equation (BW) of Figure~\ref{fig:triangle-rules}, and to prove that the interpretation of the H-box into the \dzx-calculus can be transformed back into an H-box. It is then a possibility that by doing the completeness proof natively in the ZH-calculus, we do not need the full power of \AndRule, and it suffices, for instance, to add (BW) as an axiom.

\textbf{Acknowledgements}: JvdW is supported in part by AFOSR grant FA2386-18-1-4028. The authors would like to thank Aleks Kissinger for valuable comments and support.

\bibliographystyle{eptcs}
\bibliography{main}

\appendix

\section{Proofs of derived rules}\label{sec:derived-rules}

\begin{proof}[Proof of Lemma~\ref{lem:scalarcancel}.]
    Note that by applying \StrongCompRule and \HCompRule with $n=m=0$ we already get two of the equations. For the remaining equations we simply prove them:
    \ctikzfig{scalar-rule-proof}
    \ctikzfig{HH-scalar-cancel-proof}
\end{proof}

\begin{lemma}\label{lem:negate-direct} ~
    \ctikzfig{negate-direct}
\end{lemma}
\begin{proof}~
    \ctikzfig{negate-direct-proof}
\end{proof}

\begin{proof}[Proof of Lemma~\ref{lem:zx-inspired}.] We prove all the equations in turn.
  \ctikzfig{X-spider-proof}
  \ctikzfig{X-special-proof}
  \ctikzfig{ZZ-cancel-proof}
  \ctikzfig{XX-cancel-proof}
  \ctikzfig{colour-change-proof}
  \ctikzfig{H-NOT-commute-proof}
  \ctikzfig{NOT-commute-proof}
  For the colour-reversed version of this last equation, we simply use \ColourRule.
\end{proof}

% \begin{proof}[Proof of Lemma~\ref{lem:involution}.]~
%   \ctikzfig{ZZ-cancel-proof}
%   \ctikzfig{XX-cancel-proof}
% \end{proof}

% \begin{proof}[Proof of Lemma~\ref{lem:colour-change}.]~
%   \ctikzfig{colour-change-proof}
%   \ctikzfig{H-NOT-commute-proof}
% \end{proof}

% \begin{proof}[Proof of Lemma~\ref{lem:notcopy}.]~
%   \ctikzfig{NOT-commute-proof}
%   For the colour-reversed version, we simply use \ColourRule.
% \end{proof}

% \begin{proof}[Proof of Lemma~\ref{lem:czcorrect}.]~
%   \ctikzfig{CZ-correct-proof}
% \end{proof}

\begin{proof}[Proof of Lemma~\ref{lem:statecopy}.]~
  The rules not containing any $\neg$'s all follow from a single application of \StrongCompRule or \HCompRule. The others can be shown by expanding the definition using \NotDef and \ZDef and then using some combination of the previous rules with \ColourRule and \HFuseRule, e.g. :
  \ctikzfig{state-copy-proof}
  The only one that is more complicated, is the last one, which seems to require \OrthoRule to be proven:
  \ctikzfig{both-rule-proof}
  %\ctikzfig{all-rule-proof}
\end{proof}

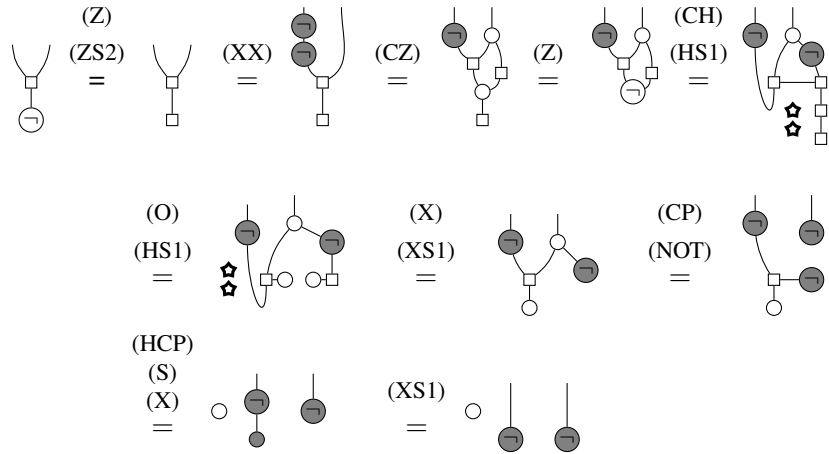
\begin{proof}[Proof of Lemma~\ref{lem:commutation}.] We simply prove all the rules in turn.
  \ctikzfig{CZ-correct-proof}
  \ctikzfig{X-Z-commute-proof}
  \ctikzfig{hopf-rule-proof}
  For the final equation in this lemma, let us first give a simple proof using \AndRule:
  \ctikzfig{dedup-simple-proof}
  It is in fact possible to derive this rule without using \AndRule. This proof is lengthy and inspired by the normal-form of Ref.~\cite{backens2018zhcalculus}. It can be found in Appendix~\ref{sec:proof-of-dedup}.
  
\end{proof}

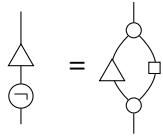
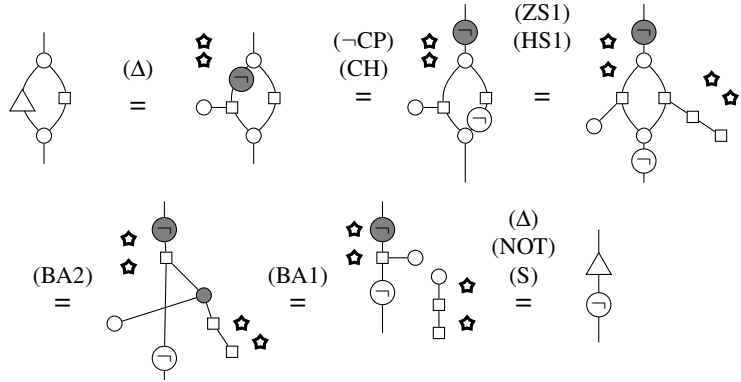
\begin{lemma}\label{lem:triangle-Z}~
  \ctikzfig{triangle-Z}
\end{lemma}
\begin{proof}~
  \ctikzfig{triangle-Z-proof}
\end{proof}
  
\begin{proof}[Proof of Lemma~\ref{lem:triangle-rules}.]~
\ctikzfig{triangle-def2-proof}
\ctikzfig{triangle-inverse-proof1}
\ctikzfig{triangle-inverse-proof2}
\end{proof}

\section{Proof of Proposition \ref{prop:dzx-to-zh}}\label{sec:proof-of-dzx-to-zh}
We need to show that whenever the \dzx-calculus proves two diagrams $D_1$ and $D_2$ to be equal, that the ZH-calculus proves $\intf{D_1}_\zh$ and $\intf{D_2}_\zh$ to be equal. In the $\dzx$-calculus diagrams are proven to be equal when some topological deformation in combination with an application of the rules transforms one diagram into the other. Since the interpretation $\intf{\cdot}_\zh$ preserves the topology, the only thing left to check then is that all the rules of the \dzx-calculus as given in Figure~\ref{fig:triangle-rules}, when transformed using $\intf{\cdot}_\zh$ are provable in the ZH-calculus. We proceed rule by rule.

\noindent \textbf{(S1)}: When the two spiders are connected by a single connection this follows easily from \SpiderRule and \ZInvolutionRule. For multiple connections we first fuse one connection, and then remove the resulting self-loops with the following:
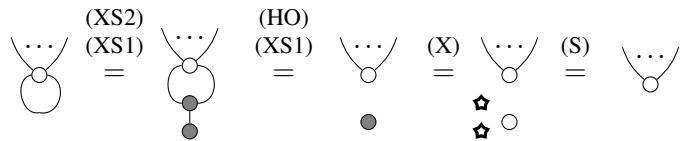
\ctikzfig{remove-self-loop}

\noindent \textbf{(S2$_g$)} and \textbf{(S2$_r$)}: These are proven respectively by \IDRule and \XIdRule.

\noindent \textbf{(IV)}:
\begin{equation*}
\intf{\input{IV.tikz}}_\zh \ \ = \ \ \input{IV-proof.tikz}
\end{equation*}

\noindent \textbf{(B1)} and \textbf{(B2)}: Trivial using \StrongCompRule and some introduction of scalars using \ScalarRule.

\noindent \textbf{(H)}: Trivial given how the interpretation $\intf{\cdot}_\zh$ is defined.

\noindent \textbf{(Z)}: Note that the legless Z-spider with a $\pi$ phase is translated to \begin{tikzpicture}
	\begin{pgfonlayer}{nodelayer}
		\node [style=white dot] (7) at (0, 0.5) {};
		\node [style=hadamard] (8) at (0, -0.5) {};
	\end{pgfonlayer}
	\begin{pgfonlayer}{edgelayer}
		\draw (7) to (8);
	\end{pgfonlayer}
\end{tikzpicture}
, and hence this rule follows easily using first \ScalarRule, and then \ZeroRule.

\noindent \textbf{(T0)}: 
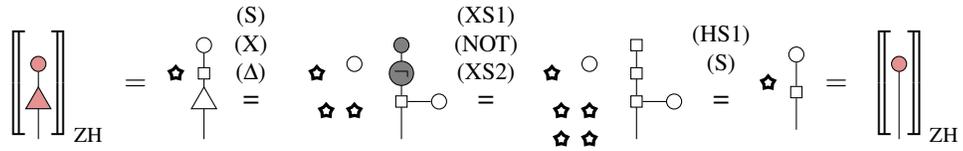
\begin{equation*}
\intf{\begin{tikzpicture}
	\begin{pgfonlayer}{nodelayer}
		\node [style=ug] (0) at (0, -0.5) {};
		\node [style=X dot] (1) at (0, 0.5) {};
		\node [style=none] (3) at (0, -1.5) {};
	\end{pgfonlayer}
	\begin{pgfonlayer}{edgelayer}
		\draw [style=none] (1) to (3.center);
	\end{pgfonlayer}
\end{tikzpicture}
}_\zh \ \ = \ \ \input{T0-left-proof.tikz} \ \ = \ \ \intf{\begin{tikzpicture}
	\begin{pgfonlayer}{nodelayer}
		\node [style=X dot] (1) at (0, 0.5) {};
		\node [style=none] (3) at (0, -1.5) {};
	\end{pgfonlayer}
	\begin{pgfonlayer}{edgelayer}
		\draw [style=none] (1) to (3.center);
	\end{pgfonlayer}
\end{tikzpicture}
}_\zh
\end{equation*}

As the translation from \dzx to ZH and back is straightforward, we will no longer write it for the following proofs.

\noindent \textbf{(BW)}: This is the only axiom for which we need to use \AndRule:
\ctikzfig{BW-proof}

\noindent \textbf{(HT)}: We could also easily prove this axiom using \AndRule, similarly to \textbf{(BW)}, but we can in fact prove it without doing that:
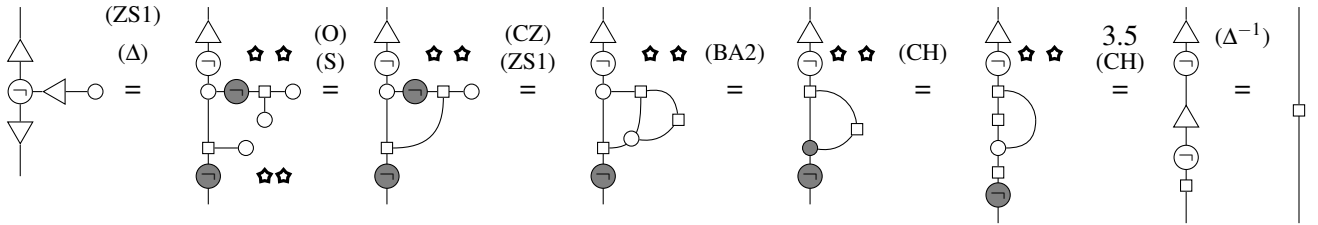
\ctikzfig{HT-proof}

\noindent \textbf{(TCX)}: This simply requires a few applications of \HCompRule:
\ctikzfig{TCX-proof}
\medskip

\noindent To prove the final rule, we will require the following lemma:
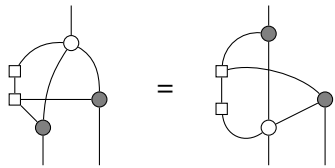
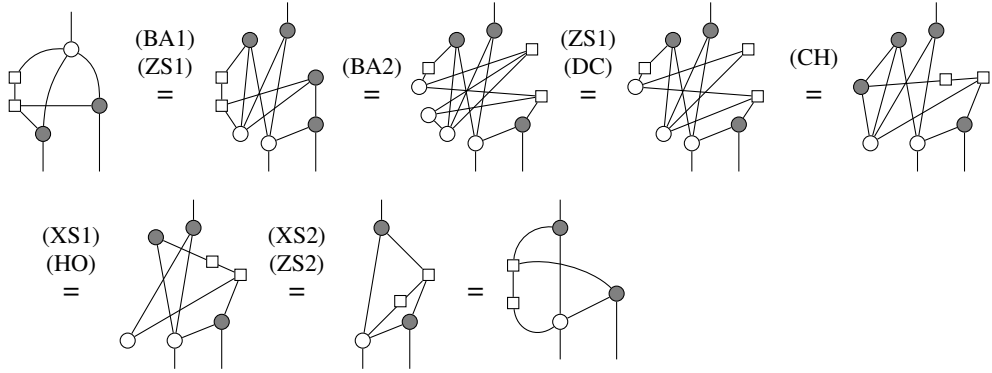
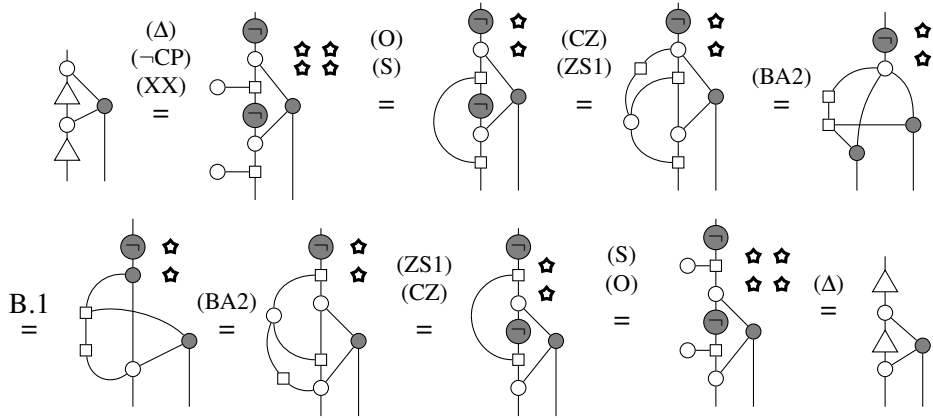
\begin{lemma}\label{lem:TW-lemma}~
\ctikzfig{TW-lemma}
\end{lemma}
\begin{proof}~
\ctikzfig{TW-lemma-proof}
\end{proof}

\noindent \textbf{(TW)}: 
\ctikzfig{TW-proof}

\pagebreak[3]
\section{Proof of Proposition \ref{prop:zh-dzx-loop}}\label{sec:proof-of-zh-dzx-loop}

We need to show for all the generators $D$ of the ZH-calculus, that $\zh \vdash \intf{\intf{D}_\dzx}_\zh = D$. Looking at the definition of these interpretation maps, we see that we immediately get equality for the identity, cup, cap, swap and Z spider. 

For the unary H-box we calculate:
$$
\intf{\intf{\hadastate{}}_\dzx}_\zh \ \ = \ \ \intf{\begin{tikzpicture}
	\begin{pgfonlayer}{nodelayer}
		\node [style=Z phase dot] (0) at (0, 0) {$\pi$};
		\node [style=none] (1) at (0, 1) {};
	\end{pgfonlayer}
	\begin{pgfonlayer}{edgelayer}
		\draw (1.center) to (0);
	\end{pgfonlayer}
\end{tikzpicture}
}_\zh =\ \  {\,\tikz{\node[style=white dot] (x) {$\neg$};\draw(x)--(0,0.75);}\,} \ \ =\ \  \hadastate{}
$$

For the binary H-box we calculate:
$$
\intf{\intf{\begin{tikzpicture}[decoration=brace]
	\begin{pgfonlayer}{nodelayer}
		\node [style=hadamard] (0) at (0, 0) {};
		\node [style=none] (1) at (0, 1) {};
		\node [style=none] (2) at (0, -1) {};
	\end{pgfonlayer}
	\begin{pgfonlayer}{edgelayer}
		\draw (2.center) to (0);
		\draw (0) to (1.center);
	\end{pgfonlayer}
\end{tikzpicture}
}_\dzx}_\zh \ \ =\ \ \intf{\input{Hadamard-normalised.tikz}}_\zh =\ \  \input{had-unnormalised.tikz}
$$

And finally, with the arity-3 H-box:
$$
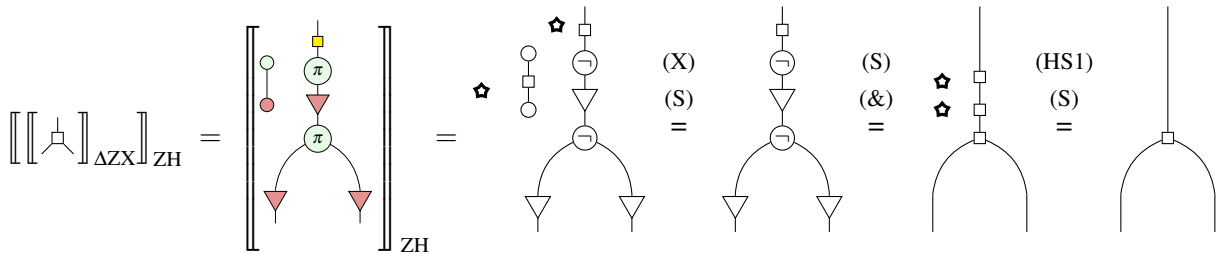
\intf{\intf{\hadamult}_\dzx}_\zh \ \  = \ \ \intf{\scalebox{0.9}{\input{AND-triangle.tikz}}}_\zh =\ \  \input{AND-triangle-normalise.tikz}
$$

\pagebreak[3]
\section{Proof of \texorpdfstring{\HadDoubleRule}{(DC)}}\label{sec:proof-of-dedup}
The normal-form construction in Ref.~\cite{backens2018zhcalculus} relies critically
on the use of arbitrary parameters on H-boxes, which is why it does not translate
to a completeness proof for the phase-free fragment. However, some parametrised H-boxes have phase-free equivalents (which can be shown using e.g.\ (A)), for instance: 
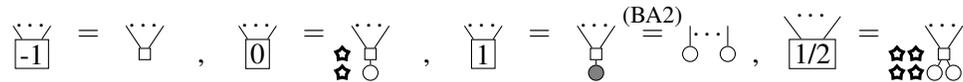
\ctikzfig{phasefree-params-demo}
While these equivalences obviously cannot be proven or even used directly
within the phase-free calculus, they show we might be able to do some normal-form inspired proofs. To do so in the case of \HadDoubleRule , we will need some lemmas. 
In this section we will use \emph{!-box} (bang-box) notation as used in Ref.~\cite{backens2018zhcalculus} to denote rules that hold for any number of wires. Furthermore, we will use \ScalarRule implicitly to introduce the required scalars for applications of the desired rules.

\pagebreak[3]
\begin{lemma}\label{lem:not-bang}~
The phase-free analogue of (M!) from Ref.~\cite{backens2018zhcalculus} holds:
\ctikzfig{not-bang-rule}
\end{lemma}
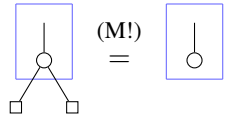
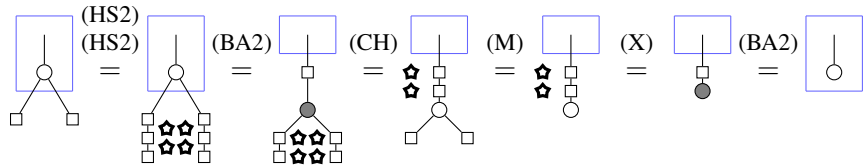
\begin{proof}~
\ctikzfig{not-bang-proof}
\end{proof}

The following is a more general version of the lemmas in Ref.~\cite[Appendix B]{backens2018zhcalculus}, and can in turn be seen as a generalisation of \OrthoRule.

\pagebreak[3]
\begin{lemma}\label{lem:splitting}~
    Instead of just wires, white spiders between an arbitrary number of H-boxes can be split using \OrthoRule:
\ctikzfig{splitting-rule}
\end{lemma}
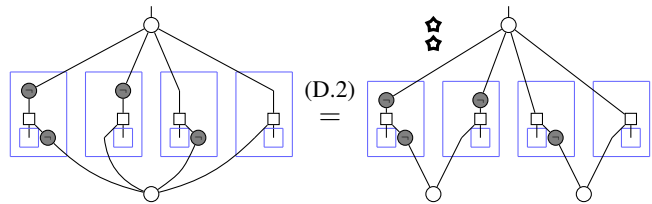
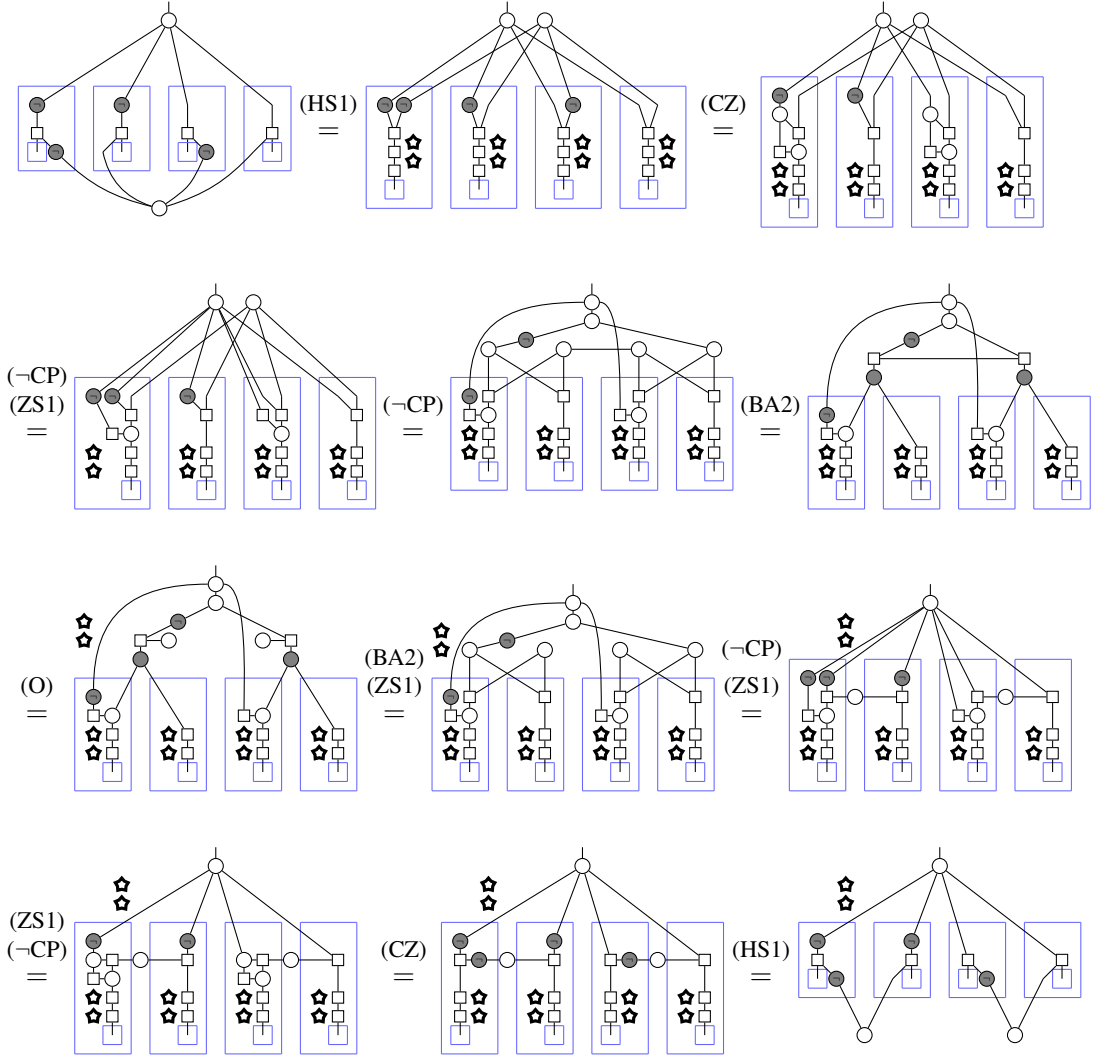
\begin{proof}~
\ctikzfig{splitting-proof}
\end{proof}

% As for practically all normal-form inspired proofs, we will need a
% disconnect lemma \SplitLemma similar to the ones demonstrated in Appendix B of
% Ref.~\cite{backens2018zhcalculus}. In this case specifically, we will also need

\pagebreak[3]
\begin{lemma}\label{lem:ortho-bang}~
    \OrthoRule can be applied on H-boxes connected by more than one wire.
%    Additionally, to make application easier, we can always drop the constraint
%    that the H-boxes to be split can only have one external wire.
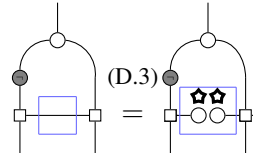
\ctikzfig{ortho-bang-rule}
\end{lemma}
\begin{proof}~
    This is obviously true with 0 wires, and 1 wire is just \OrthoRule. By induction: \\
    \scalebox{0.9}{\input{ortho-bang-altproof.tikz}}
\end{proof}

\pagebreak[3]
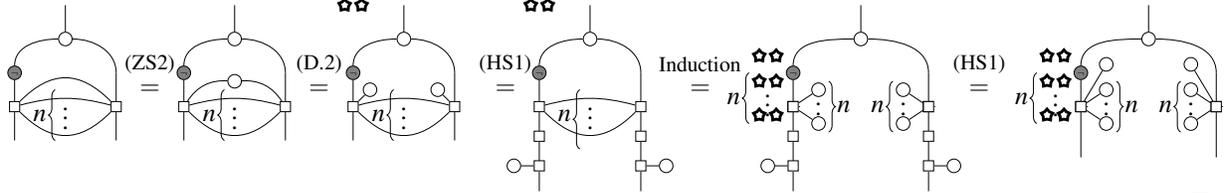
\begin{lemma}\label{lem:dyadic-intro}~
    A phase-free analogue of (I!) from Ref.~\cite{backens2018zhcalculus} for a limited family of diagrams holds:
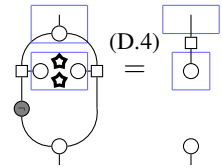
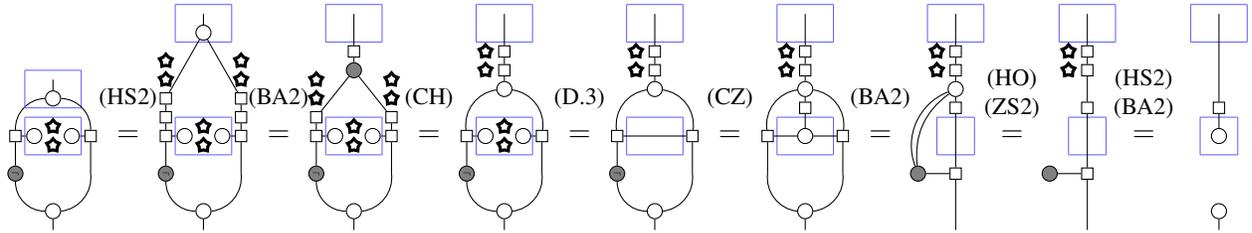
\ctikzfig{dyadic-comp-intro}
\end{lemma}
\begin{proof}~
\ctikzfig{dyadic-comp-intro-proof}
\end{proof}

\pagebreak[3]
\begin{lemma}\label{lem:id-schur}~
    The identity can be brought to the following diagram resembling its normal-form:
\ctikzfig{id-schur}
\end{lemma}
\begin{proof}~
\ctikzfig{id-schur-proof}
\end{proof}

And finally, we can prove \HadDoubleRule:

\begin{lemma}\label{lem:normal-dedup}~
\ctikzfig{normal-dedup}
\end{lemma}
\begin{proof}~
\ctikzfig{dedup-normal-condensed}
\end{proof}

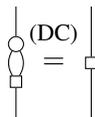
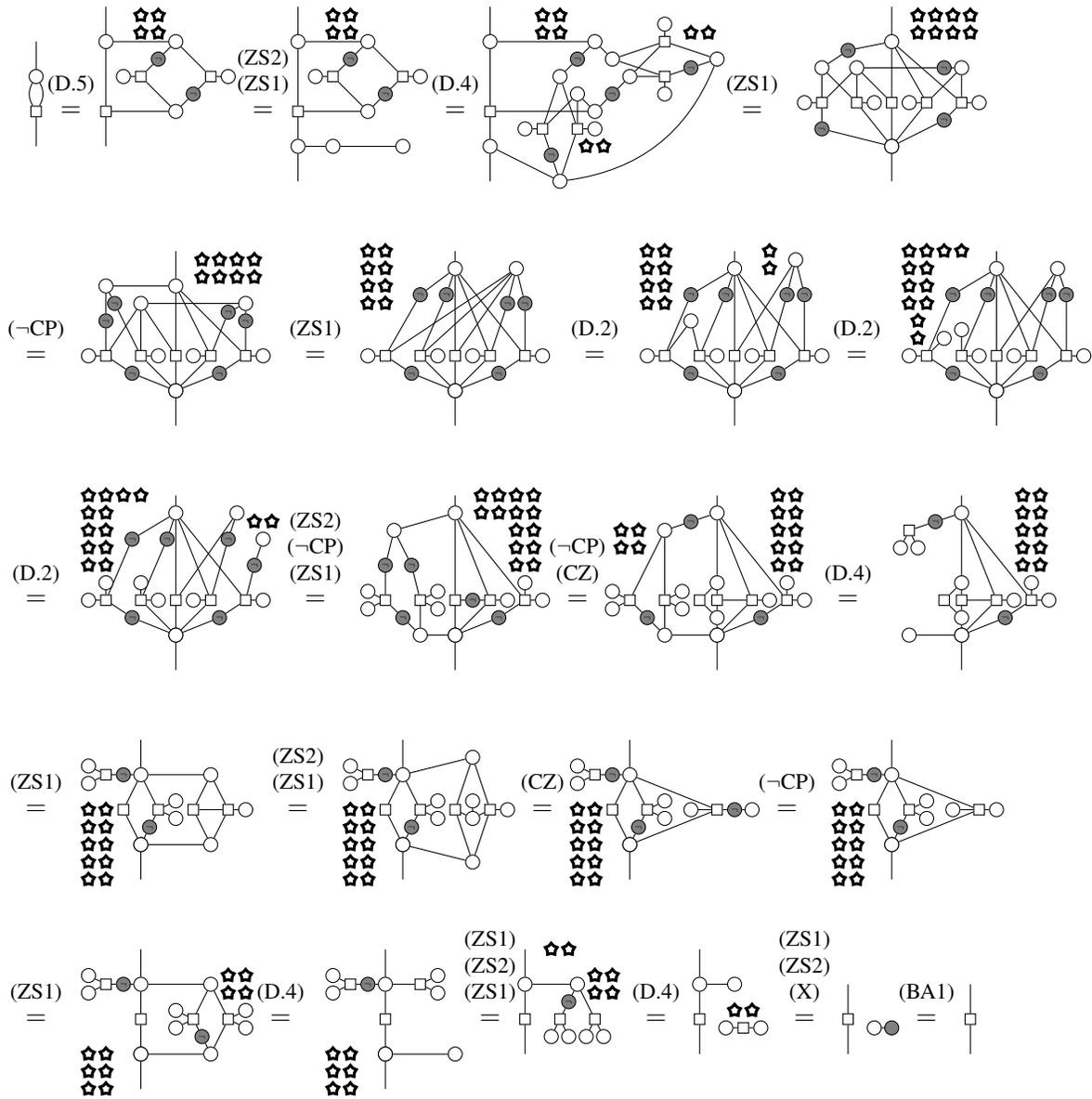
\end{document}